\documentclass[final, 11pt]{article}

\usepackage{amsmath,amssymb,amsfonts}
\usepackage{mathtools}
\usepackage{epsfig}
\usepackage{latexsym,nicefrac,bbm}
\usepackage{xspace}
\usepackage{color,fancybox,graphicx,url,subfigure}
\usepackage{enumitem}
\usepackage{fullpage}
\usepackage{mathabx}
\usepackage{tabularx}
\usepackage{mdframed}
\usepackage{picture}
\usepackage{bm}
\usepackage{pgf,tikz}
\usetikzlibrary{arrows}

\usepackage[bookmarks,colorlinks,breaklinks]{hyperref}  
\hypersetup{linkcolor=blue,citecolor=blue,filecolor=blue,urlcolor=blue} 
\newcommand{\lref}[2][]{\hyperref[#2]{#1~\ref*{#2}}}
\renewcommand{\eqref}[1]{\hyperref[#1]{(\ref*{#1})}}

\def\showauthornotes{1}

\def\showdraftbox{0}

\newtheorem{theorem}{Theorem}[section]

\newtheorem{lemma}[theorem]{Lemma}

\def\FullBox{\hbox{\vrule width 6pt height 6pt depth 0pt}}

\def\qed{\ifmmode\qquad\FullBox\else{\unskip\nobreak\hfil
\penalty50\hskip1em\null\nobreak\hfil\FullBox
\parfillskip=0pt\finalhyphendemerits=0\endgraf}\fi}

\def\qedsketch{\ifmmode\Box\else{\unskip\nobreak\hfil
\penalty50\hskip1em\null\nobreak\hfil$\Box$
\parfillskip=0pt\finalhyphendemerits=0\endgraf}\fi}

\newenvironment{proof}{\begin{trivlist} \item {\bf Proof:~~}}
   {\qed\end{trivlist}}

\newcommand\R{\mathbb R}

\newcommand{\marginlabel}[1]%
{\mbox{}\marginpar{\it{\raggedleft\hspace{0pt}#1}}}

\definecolor{Mygray}{gray}{0.8}

 \ifcsname ifcommentflag\endcsname\else
  \expandafter\let\csname ifcommentflag\expandafter\endcsname
                  \csname iffalse\endcsname
\fi

\ifnum\showauthornotes=1
\else
\fi

\ifnum\showauthornotes=1
\newcommand{\Authornote}[2]{{\sf\small\color{red}{[#1: #2]}}}
\newcommand{\Authoredit}[2]{{\sf\small\color{red}{[#1]}\color{blue}{#2}}}
\newcommand{\Authorcomment}[2]{{\sf \small\color{Mygray}{[#1: #2]}}}
\newcommand{\Authorfnote}[2]{\footnote{\color{red}{#1: #2}}}
\newcommand{\Authorfixme}[1]{\Authornote{#1}{\textbf{??}}}
\newcommand{\Authormarginmark}[1]{\marginpar{\textcolor{red}{\fbox{
#1:!}}}}
\else
\newcommand{\Authornote}[2]{}
\newcommand{\Authoredit}[2]{}
\newcommand{\Authorcomment}[2]{}
\newcommand{\Authorfnote}[2]{}
\newcommand{\Authorfixme}[1]{}
\newcommand{\Authormarginmark}[1]{}
\fi

\newlength{\pgmtab}  
 {
	\begin{enumerate}}{\end{enumerate}}

\newcounter{lecnum}

\newlength{\tpush}
\ifnum\showdraftbox=1

\else

\fi

% ...

\newcommand{\F}{{\mathbb F}}

\newcommand{\Q}{{\mathbb Q}}

\newcommand{\NP}{{\mathsf {NP}}}
\newcommand{\PTIME}{{\mathsf {P}}}
\newcommand{\PSPACE}{{\mathsf {PSPACE}}}
\newcommand{\sgn}{{\mathsf {sign}}}
\newcommand{\minrank}{{\mathsf {minrank}}}
\newcommand{\rank}{{\mathsf {rank}}}
\newcommand{\mat}{{\mathsf {Mat}}}
\newcommand{\spn}{{\mathsf {span}}}

\newcommand{\spm}{sign pattern matrix}
\newcommand{\spms}{sign pattern matrices}
\newcommand{\gspm}{generalized sign pattern matrix}
\newcommand{\gspms}{generalized sign pattern matrices}

\newtheorem{cor}{Corollary}
\newtheorem{df}{Definition}
\newtheorem{obs}{Observation}

\begin{document}

\title{The complexity of computing the minimum rank of a sign pattern matrix
\footnote{After this paper as posted online, we learnt about the paper~\cite{Basri09}, which gave hardness results and algorithms for
sign rank several years earlier. On the algorithmic side, \cite{Basri09} gave a proof of Theorem~\ref{thm:algo}, which is closely related to our proof. For the hardness results, \cite{Basri09} also used Shor's result~\cite{shor} to prove Theorem~\ref{thm:main} on generalized sign pattern matrices, and indicated that
similar ideas could be used to prove Theorem~\ref{thm:strictmain} on sign pattern matrices. Our paper gives a detailed proof of Theorem~\ref{thm:strictmain}, using a structural result on uniform pseudoline arrangments, Lemma~\ref{lemma:suff}. Hence we believe our paper may still be of
some interest. }}

\author{Amey Bhangale
\thanks{Department of Computer Science. Rutgers University. Research supported in part by NSF grant CCF-1253886. {\tt amey.bhangale@rutgers.edu}}
\and 
Swastik Kopparty \thanks{Department of Mathematics \& Department of
  Computer Science. Rutgers University. Research supported in part by a Sloan Fellowship and NSF grant CCF-1253886. {\tt swastik.kopparty@rutgers.edu}}
}
\maketitle

\begin{abstract}
We show that computing the minimum rank of a sign pattern matrix is NP hard. 
Our proof is based on a simple but useful connection between 
minimum ranks of sign pattern matrices and the stretchability problem
for pseudolines arrangements. In fact, our hardness result shows that it is already hard to determine
if the minimum rank of a sign pattern matrix is $\leq 3$. We complement this by giving a polynomial time algorithm for determining if a given sign pattern
matrix has minimum rank $\leq 2$.

Our result answers one of the open problems from \cite{Linialetal} [{\em Combinatorica}, 27(4):439--463, 2007].
\end{abstract}


\section{Introduction}

We study the complexity of computing the minimum rank 
of a sign pattern matrix.  The minimum rank
of a sign pattern matrix is a basic quantity that arises naturally 
in a number of settings, and has been extensively studied.
 Our main result shows that it is $\NP$ hard to compute the minimum
rank of a sign pattern matrix. In fact, we show that it is even $\NP$ hard
to distinguish between sign pattern matrices with minimum rank $= 3$ from 
sign pattern matrices with minimum rank $ > 3$. Our proof proceeds
via a simple but powerful connection between the rank of a sign pattern matrix
and a classical problem in oriented matroid theory and computational
geometry: the stretchability problem for pseudoline arrangements.
We complement our hardness results by giving a polynomial time algorithm
to detect if a given sign pattern matrix has minimum rank $\leq 2$.

We begin by setting up some notation. A {\spm}
is a matrix whose entries come from the set $\{+, - \}$.
If the entries come from $\{+, -, 0\}$, we call the matrix a 
{\em generalized} {\spm}.
Given a real matrix $A$, one can consider its associated
generalized sign pattern matrix $\sgn(A)$ (obtained by taking the sign of each entry).
The minimum rank of a generalized sign pattern matrix $S$, denoted $\minrank(S)$,
is defined to be the minimum, over all real matrices $A$ with $\sgn(A) = S$,
of $\rank(A)$.

The minimum rank of a sign pattern matrix has been
widely studied in the discrete mathematics and theoretical computer
science communities. On the discrete mathematics side, there has
been much recent work on the relationship between the minimum rank
of a sign pattern matrix and the combinatorial
structure of associated bipartite graph. 
The recent workshop at the American Institute for Mathematics \cite{AIM}
provides and excellent overview of some of the
facets of this line of work.
On the theoretical computer science side, the minimum rank
of a sign pattern matrix gives an exact characterization of
the randomized unbounded-error communication complexity (i.e., the communication
complexity analogue of $\mathsf{PP}$). This has motivated a number of interesting
problems and results on minimum ranks of sign pattern matrices.
Rank minimization of other classes of matrices has also been studied
in the context of machine learning and computational linear algebra.

We now list some of the basic results and problems on sign pattern matrices
arising from these two communities.

\begin{enumerate}
\item As observed by Berman et al~\cite{Hogben}, computing the minimum rank of a sign pattern matrix
easily reduces to checking satisfiability a system of real polynomial equalities and inequalities.
This problem is known as ``the existential first order theory of the reals", and its complexity
is known to lie somewhere between $\NP$ and $\PSPACE$~\cite{Renegar}.

\item For a subfield $\F \subseteq \R$ and a sign pattern matrix $S$,
denote by $\minrank_{\F}(S)$ the minimum rank over all matrices
$A$ with entries in $\F$ such that $\sgn(A) = S$. Arav et al\cite{Aravetal}
showed that if $S$ is a {\gspm}
such that $\minrank_{\R}(S) \in \{1, 2\}$, then
$\minrank_{\F}(S) = \minrank_{\R}(S)$.  Berman et al~\cite{Hogben}
and Kopparty and Rao~\cite{KR-minrank} showed that in general,
the minimum rank depends on the field: there exist
$S$ with $\minrank_{\R}(S) = 3$ but $\minrank_{\Q}(S) > 3$.
In fact, the behavior of $\minrank_{\F}(S)$ depends intricately on
$\F$ (see~\cite{KR-minrank}).

For {\spms} $S$, we always have
$\minrank_{\R}(S) = \minrank_{\F}(S)$~\cite{Aravetal}.

\item
A basic result of Alon et al, \cite{Alonetal} shows that the minimum rank of 
every $n \times n$ {\spm} is at most $\frac{n}{2} ( 1+ o(1))$.
\cite{Alonetal} also show that there are $n \times n$ {\spms} with 
minimum rank at least $\frac{n}{32}$. For generalized sign pattern
$n \times n$ matrices, the minimum rank can be an arbitrary integer in $[0, n]$.

\item
The closely related problem of computing the maximum
rank of a {\gspm} has a simple polynomial time algorithm;
the maximum rank of a {\gspm} $S$ 
simply equals the size of the largest matching $(r_1, c_1), (r_2, c_2), \ldots $
between rows and columns, such that for each $i$, the entry $S_{r_i c_i}$ is
nonzero. This immediately reduces to finding a maximum matching in a bipartite graph.

\item
Buss, Frandsen and Shallit~\cite{BFS} studied the complexity of a
number of linear algebraic problems on matrices. Given a matrix
with each entry being either a variable or a constant, and a field $\F$, they considered the
complexity of computing the minimum rank and maximum rank (and some variants) 
as we vary over all substitutions of the variables with values from $\F$.
Over the field $\mathbb R$, \cite{BFS} show that the problem of computing
maximum rank lies in $\mathsf{RP}$ (see also Lovasz~\cite{Lovasz}),
and the problem of computing minimum rank is equivalent to deciding the
first order theory of the reals (and is thus $\mathsf{NP}$-hard).

Computing the minimum rank of a {\gspm} can be posed in this framework
by considering a matrix whose entries are variables and $0$'s, with two further 
restrictions:
\begin{enumerate}
\item[(i)] to each variable we associate a sign which forces that variable to be either strictly positive
or strictly negative,
\item[(ii)] more importantly, each variable appears at most once in the matrix.
\end{enumerate}
Indeed, \cite{BFS} also studied condition (ii), and were able to show some improved upper
bounds for some linear algebraic problems with this condition. However, as they noted,
their $\NP$-hardness results did not apply with condition (ii).

\item A fundamental result of Paturi et al~\cite{paturi} exactly characterizes the 
randomized unbounded-error two-party communication complexity of a function
in terms of the minimum rank of an associated {\spm}.
This motivated some beautiful work in complexity theory on lower bounds
for the minimum rank of a {\spm} (a.k.a. ``sign rank")
~\cite{Alonetal, RS, Forster, Linialetal}.

\item The best known lower bound on the minimum rank of an explicit 
$n \times n$ {\spm} is $\Omega(\sqrt{n})$. This
result is due to Forster~\cite{Forster}, who proved this by means of a new
general lower bound on the minimum rank of a {\spm} in terms of
the spectrum of its associated $\{ +1, -1 \}$ matrix.

\item Forster's lower bound was generalized by Linial, Mendelson, Shraibman and Shechtman~\cite{Linialetal}.
They studied various complexity measures of sign matrices, and proved relationships between them.
In particular, they proved an upper bound on the minimum rank of a {\spm} in terms
of its margin complexity, a notion that arises in machine learning.

\item A recent new approach to matrix multiplication, due to Cohn and Umans~\cite{CU}, is based
on finding tensors with low {\em support rank}. The support rank of a zero-nonzero pattern tensor
is the minimum rank amongst all tensors with the given zero-nonzero pattern.
Via an example, they noted that the concept of support rank is already quite intricate
for two-dimensional tensors  (i.e., matrices) by giving an example of a zero-nonzero pattern matrix whose
minimum rank differs wildly from the rank of its associated 0-1 matrix.

Our techniques can also be used to show that computing the support rank of a zero-nonzero matrix is $\NP$ hard.
This answers a question raised by Umans in a recent talk. In fact, the $\NP$ hardness already appears
at support rank $3$.

\item The complexity of a similarly themed problem, that of computing the nonnegative rank of a matrix,
was recently studied by Arora, Ge and Moitra~\cite{AGM,M13}. There it was shown that for any constant $r$,
matrices of non-negative rank $\leq r$ can be recognized in polynomial time. This is in sharp contrast to
our result for minimum rank of a {\spm}, where the $\NP$ hardness appears already at minimum rank $3$.

\end{enumerate}

\subsection{Results}

We now state our main results on the hardness of computing the minimum rank of a {\spm}.

Our first (and simpler) main theorem deals with {\gspms}.
\begin{theorem}
\label{thm:main}
The following problem is equivalent to the existential first order theory of the reals:
Given a {\gspm} $S$, decide whether $\minrank(S) \leq 3$. In particular, this problem is $\NP$ hard.
\end{theorem}

This has the following immediate corollary, which rules out additive $O(1)$ approximations
to the minimum rank of a {\gspm}, unless $\PTIME = \NP$.

\begin{cor}
For every integer $k \geq 1$, following problem is equivalent to the existential first order theory of the reals:
Given  a {\gspm} $S$ and an integer $r$, distinguish between the case that $\minrank(S) \leq r$ and
the case that $\minrank(S) \geq r + k$.

In particular, this problem is $\NP$ hard.
\end{cor}

Our second main theorem concerns {\spms}.

\begin{theorem}
\label{thm:strictmain}
The following problem is $\NP$ hard: Given a {\spm} $S$,
decide whether $\minrank(S) \leq 3$.
\end{theorem}

Here we can only show $\NP$ hardness. We also do not get the analogous corollary:
we do not even know how to rule out (conditional on $\PTIME \neq \NP$)
the existence of a polynomial time algorithm that computes 
the minimum rank of a {\spm} within an additive $2$.

Finally, we show that computing the minimum rank of a {\gspm}
is easy when the minimum rank is $\leq 2$.

\begin{theorem}
\label{thm:algo}
For every $r \in \{0, 1, 2\}$, the following problem can be solved in polynomial time:
Given a {\gspm} $S$, decide whether $\minrank(S) \leq r$.
\end{theorem}

Some remarks:
\begin{enumerate}
\item Our $\NP$ hardness results use a beautiful
and highly nontrivial hardness result for a basic problem 
in computational geometry: {\em the stretchability of pseudoline arrangements}. This
problem (to be defined later) was shown to be $\NP$ hard
by Shor~\cite{shor}, based on some ingenious projective geometry constructions.
Shor also noted that hardness for the stretchability problem for pseudoline arrangements
could be derived from Mnev's universality theorem for oriented matroids \cite{mnev},
by observing that the steps of Mnev's proof can be executed in polynomial time.

\item The fact that detecting if $\minrank(S) \leq 2$ is
easy and detecting if $\minrank(S) \leq 3$ is hard is analogous
to the fact that the statement $\minrank_{\F}(S) \leq 2$ does not depend
on the field $\F$, while the statement $\minrank_{\F}(S) \leq 3$ does
depend on the field. In fact, there are some common themes in our proof:
the field-dependence of $\minrank$ as proved in~\cite{KR-minrank}, depends on a projective
geometry construction known as the ``von Staudt algebra of throws"; this
same construction appears as a crucial tool in the proof of Mnev's universality
theorem, and thus plays a role in our hardness result.

Note, however, that for {\spm} $S$, $\minrank_{\F}(S)$
is independent of the field. Nevertheless, Theorem~\ref{thm:strictmain}
shows that it is still hard to compute the minimum rank of a {\spm}. 
The proof of Theorem~\ref{thm:strictmain} involves some further ideas beyond
what goes into the proof of Theorem~\ref{thm:main}.

\end{enumerate}

\subsection{Proof outline}

We now give an overview of the proofs.

Let us work on the real projective plane $\mathbb {RP}^2$.
A pseudoline is a non-self-intersecting continuous curve in $\mathbb {RP}^2$
that intersects the line at infinity in exactly 1 point. A pseudoline arrangement
is a collection of pseudolines such that any two pseudolines in the collection
intersect in exactly $1$ point. Two pseudoline arrangements are called 
{\em equivalent} if there is a homeomorphism of $\mathbb {RP}^2$ which induces a bijection
on the pseudolines in the pseudoline arrangements. A pseudoline arrangement is
called {\em stretchable} if it is equivalent to a pseudoline arrangement where each 
pseudoline is in fact a straight line in $\mathbb R^2$.

Our hardness results for computing the  minimum rank of {\spm} and {\gspm}
are based on the following hardness results for determining stretchability of pseudoline
arrangements:
\begin{theorem}[{\cite{mnev, shor}}]
The following problem is equivalent to the existential first order theory of the reals:
Given a pseudoline arrangement, decide whether it is stretchable.
\end{theorem}
A uniform pseudoline arrangement is a pseudoline arrangement where no three pseudolines meet at a point.
\begin{theorem}[\cite{shor}]
The following problem is $\NP$ hard: 
Given a uniform pseudoline arrangement, decide whether it is stretchable.
\end{theorem}

To prove Theorem~\ref{thm:main}, we give a simple reduction of the stretchability problem for pseudoline arrangements to
the problem of determining whether the minimum rank of a {\gspm}
is $\leq 3$. The main idea is the following: given a collection of points $(x_1, y_1), \ldots, (x_n, y_n) \in \R^2$,
and a collection of lines $L_1 \equiv \{ a_1 X + b_1 Y + c_1 = 0 \}, \ldots, L_m \equiv \{ a_m X + b_m Y + c_m = 0 \}$,
consider the matrices $P_{n\times 3}$ whose row $i$ is $(x_i,y_i,1)$ and $L_{3\times m}$ whose $j^{th}$ column is $(a_j,b_j,c_j)$. Note that the matrix $A = P \cdot L$ has rank $\leq 3$,
and its sign pattern matrix $\sgn(A)$ precisely
captures the relative position information between the points and the lines: it tells us which points lie on which sides of each line,
and which points lie on which lines.
The reverse is also true: given a {\gspm} $S$, it has minimum rank $\leq 3$ if and only if one can realize
a set of points and a set of lines with certain prescribed relative positions (depending on $S$).
Finally, we show how the stretchability problem for pseudoline arrangements (which is a problem of realizing a set of lines with
prescribed incidence-like properties) can be reduced to the problem of realizing point-line configurations with prescribed relative positions. This is done using some basic oriented matroid theory.

To prove Theorem~\ref{thm:strictmain}, we need to end up with a sign pattern matrix without any $0$'s.
Following the above ideas, we need to reduce some $\NP$-hard problem to the problem of realizing a point-line configuration with prescribed relative positions of points and lines, {\em where no point of the configuration lies on any line of the configuration}. It turns out that the simple reduction
described above for Theorem~\ref{thm:main} does not produce such point-line configurations. Instead, we give a slightly more involved reduction which starts from the stretchability problem for {\em uniform} pseudoline arrangements, and uses some properties of their associated oriented matroids. This leads to Theorem~\ref{thm:strictmain}.

We now outline our algorithmic results.
We first give a polynomial time algorithm to check if the minimum rank of a {\spm} is at most $2$. Given a $m\times n$  {\spm}, we construct a family of $O(n)$ subsets of $[m]$, and  show that this family satisfies a ``chain'' like property if and only if the given {\spm} has minimum rank at most $2$. Checking the ``chain'' property turns out to be in polynomial time. We then extend this result to {\gspm} by reducing the problem for {\gspm} to the problem for {\spm}.


\section{Hardness results}

Our reduction from the problem of stretchability of pseudoline arrangements to the minimum rank problem will involve
the notion of an {\em oriented matroid}.
We begin with some preliminaries on pseudoline arrangements and oriented matroids.

\subsection{Preliminaries}

We will use some basic properties of oriented matroids. Oriented matroids are
abstract combinatorial objects closely related to matroids. Here we will use an axiomatization of oriented matroids based on ``covectors" (this axiomatization  is not related to any of the standard axiomatizations of matroids). See \cite{handbook_dcg_6, orientedmatroid} for an introduction to the area. 

\begin{df}
{\bf Oriented matroid:} An oriented matroid is a pair $\mathcal{M} = (E,\mathcal{L})$ where $E$ is a ground set, and $\mathcal{L} \subseteq \{0,-,+\}^{E}$,
such that $\mathcal{L}$ satisfies the axioms (CV0)-(CV4) given in \cite[p. 118]{handbook_dcg_6}. $\mathcal L$ is called the set of covectors of $\mathcal M$.
\end{df}

\begin{df}
{\bf Reorientation of a Matroid:} Let $\mathcal{M} = (E,\mathcal{L})$ be an oriented matroid. The reorientation of $\mathcal{M}$ with respect to set $A\subseteq E$ is $\mathcal{M}' = (E,\mathcal{L}')$, such that for each $v\in \mathcal{L}$, we have a corresponding $v' \in \mathcal{L}'$ with the property that $v'_i = -v_i$ if $i\in A$ and $v'_i = v_i$ otherwise.
\end{df}

\begin{df}
{\bf Isomorphism of Oriented Matroids:} Two oriented matroids $(E_1, \mathcal L_1)$ and $(E_2, \mathcal L_2)$ are called isomorphic if 
$E_1 = E_2$ and there is a reorientation $(E_1, \mathcal L_1')$ of $(E_1, \mathcal L_1)$ such that $\mathcal L_1' = \mathcal L_2$.
\end{df}
To every pseudoline arrangement $\mathcal{P}$, we can associate an oriented matroid $\mathcal{OM}(\mathcal P)$ as follows:

\vspace{20pt}

\begin{figure}[h!]
\label{fig:pa}
\centering
\begin{tikzpicture}[line cap=round,line join=round,x=0.65cm,y=0.65cm]

\draw (0,5.5) node[anchor=north west] {$l_1$};
\draw (1,5)-- (13,1);
\draw [->, blue, >=triangle 45] (1.2,4.95) -- (1.5,5.6);

\draw (0,1.5) node[anchor=north west] {$l_4$};
\draw plot [smooth ] coordinates {(1,1) (5,3) (13,4)};
\draw [->, blue, >=triangle 45] (1.3,1.2) -- (1.5,0.3);

\draw (0,2.5) node[anchor=north west] {$l_3$};
\draw plot [smooth ] coordinates {(1,2) (6,1.5) (11,4.5) (13,5)};
\draw [->, blue, >=triangle 45] (1.3,1.9) -- (1.5,2.7);

\draw (0,3.5) node[anchor=north west] {$l_2$};
\draw plot [smooth ] coordinates {(1,3) (6,5) (13,2)};
\draw [->, blue, >=triangle 45] (1.2,3.1) -- (1,3.9);

\draw[fill] (9.7,3.67) circle(0.1);
\draw [->] (9.7,3.67) -- (10,5.6); \draw (9.2, 6.6) node[anchor=north west] {$(+$ $0$ $0$ $0)$};

\draw[fill] (8,1) circle(0.1); \draw (8,1) node[anchor=north west] {$(- - - +)$};

\draw[fill] (6,5) circle(0.1); 
\draw [->] (6,5) -- (6.5, 6); \draw (5.5,6.8) node[anchor=north west] {$(+$ $0 + -)$};

\end{tikzpicture}
\caption{Pseudoline arrangement and the associated sign patterns of some points. The $+$ side in the orientation is denoted by blue arrow.}
\end{figure}
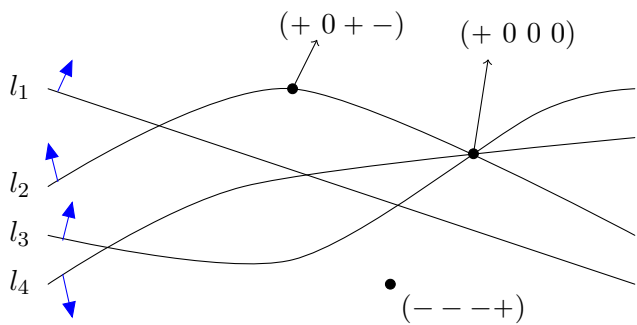
Consider the cell complex obtained by the subdivision of $\mathbb{R}^2$ by the pseudoline arrangement $\mathcal P$. Each pseudoline partitions $\mathbb R^2$ into two regions;
we arbitrarily choose one side to be the $+$ side of that pseudoline and the other side to be $-$ side of that pseudoline.
Let $E$ be the set of pseudolines in $\mathcal P$. For a point $p$ in $\mathbb R^2$, we now get a sign vector $s_p \in \{0, -, + \}^E$:
the $\ell$th coordinate of $s_p$ equals $+/ -/ 0$ depending on whether $p$ is on the $+$ side/$-$ side/on the pseudoline $\ell$.
$s_p$ is called the covector of $p$ (see figure~\ref{fig:pa}). We take $\mathcal L$ to equal the set of all covectors $s_p$, as $p$ varies in $\mathbb{R}^2$, along with
the all $0$s vector. Then $\mathcal M = (E, \mathcal L)$ is an oriented matroid, and is called the oriented matroid
associated with $\mathcal P$, and is denoted $\mathcal{OM}(\mathcal P)$.

Note the oriented matroid depends on the arbitrary choice of $+$ and $-$ sides of the pseudolines.
The oriented matroid $\mathcal{OM}(\mathcal P)$ is only uniquely defined up to isomorphism.

The covectors $s_p$ obtained above can be classified into $3$ categories:
$\mathcal{C}_0(\mathcal{P})$, $\mathcal{C}_1(\mathcal{P})$ and $\mathcal{C}_2(\mathcal{P})$, depending on whether $p$
is in a $0$-cell (point), $1$-cell (curve segment) or $2$-cell (open 2-dimensional region) of the cell complex cut out by
$\mathcal P$ (note that all $p$ in the same cell have the same covector $s_p$). Furthermore, given a covector $s_p$,
it is easy to determine which $\mathcal{C}_i(\mathcal{P})$ it belongs to: if the number of $0$s in $s_p$ is $\geq 2$, then
it belongs to $\mathcal C_0(\mathcal P)$, if the number of $0$s in $s_p$ equals $1$, then it belongs to $\mathcal C_1(\mathcal P)$,
 and otherwise it belongs to $\mathcal C_2(\mathcal P)$.

The equivalence between two pseudoline arrangements can be viewed as a purely combinatorial condition in terms of their oriented matroids.
This is given by the easy direction of Folkman-Lawrence Topological Representation Theorem.

\begin{theorem}
[\cite{folkman-lawrence}, {\bf Topological Representation Theorem}]
\label{thm:trt}
Two pseudoline arrangements are equivalent if and only if their oriented matroids are isomorphic. 
\end{theorem}

Thus, a given pseudoline arrangement $\mathcal P$ is stretchable if and only if there is a line arrangement whose oriented matroid can be reoriented to 
equal the oriented matroid of $\mathcal P$.

\paragraph{Representation:} One may wonder in what format pseudoline arrangements are represented.
Since we only care about pseudoline arrangements up to equivalence, we need not precisely describe the 
pseudolines appearing in the arrangement, but can work with some coarser representation. One natural candidate
is to work with ``wiring diagrams" given by polygonal curves. For us, it will be sufficient to represent
a pseudoline arrangement by the oriented matroid associated with it; note that Theorem~\ref{thm:trt} implies
that this data is sufficient to capture equivalence.

For our purposes, it suffices to note that the oriented matroid of the pseudoline arrangement produced
by the Shor's reductions~\cite{shor} can be computed in polynomial time. This can be seen by inspecting the proof.

\subsection{Proof of Theorem~\ref{thm:main}}

We now give our hardness reduction. Given a pseudoline arrangement, we construct a {\gspm} 
such that minimum rank of the constructed {\gspm} is at most 3 if and only if the pseudoline
arrangement we started with was stretchable.

For any $\mathcal{L} \subseteq \{0,-,+\}^n$, let $\mat(\mathcal{L})$ be the {\gspm} whose rows correspond to the sign vectors in $\mathcal{L}$ in arbitrary order.

\begin{lemma}
\label{lemma:main}
Let $\mathcal{P}$ be a pseudoline arrangement. Let $\mathcal {OM}(\mathcal P) = (E, \mathcal L)$ be its oriented matroid.
Let $S$ be a {\gspm} obtained by augmenting a column of all $+$s to $\mat(\mathcal{L})$.
Then $\mathcal{P}$ is stretchable if and only if $\minrank(S) \leq 3$.
\end{lemma}
\begin{proof}
Let $m = |\mathcal{C}_0(\mathcal{P})| +  |\mathcal{C}_1(\mathcal{P})| + |\mathcal{C}_2(\mathcal{P})|$ (the total number of covectors of $\mathcal P$),
and let $n = |E|$. Note that $S$ is a $m \times (n+1)$ {\gspm}.

Let us first prove that if the given pseudoline arrangement is stretchable, then $\minrank(S) \leq 3$. If the given pseudoline arrangement is stretchable then there is a line arrangement $\mathcal L$ such that a suitable reorientation of the oriented matroid associated with this line arrangement equals $\mathcal {OM}(\mathcal P)$.
Suppose the lines in $\mathcal L$ are given by the equations $a_j X + b_j Y + c_j = 0$, for $j \in [n+1]$. Let $(x_1, y_1), \ldots, (x_m, y_m)$ be points
in $\mathbb R^2$ whose covectors w.r.t. $\mathcal L$ equal all the possible covectors of $\mathcal L$.

Now construct matrices $L_{m\times 3}$ and $R_{3\times (n+1)}$ as follows. The $i$th row of $L$ equals $(x_i,y_i,1)$.
The $j$th column of $R$ equals $(a_j,b_j,c_j)$. Now consider $A=L \cdot R$, and note that $\rank(A) \leq 3$. Note that $\sgn(A_{ij})$ equals $+/-/0$ depending
on whether $(x_i, y_i)$ lies on the $+$ side/$-$ side/on the $j$th line. Thus the $\sgn$s of the rows of $A$ equal
the covectors of $\mathcal L$. We know that there is a reorientation of $\mathcal {OM}(\mathcal P)$ which equals $\mathcal {OM}(\mathcal P)$;
this means that flipping the signs of some rows of $A$ will give us a matrix $A'$ with $\sgn(A') = S$. Thus $S$ has minimum rank at most $3$.

In the other direction, suppose that $\minrank(S) \leq 3$. Then there exists some matrix $A \in \mathbb{R}^{m \times (n+1)}$ having rank at most $3$ and $\sgn(A) = S$. By dividing row $i$ of $A$ with $A_{(i,n+1)}$ for all $i$ can make the last column the all $1$s vector. Since the column rank of $A$ is at most $3$, there exist at most three linearly independent vectors such that every column vector of $A$ can be represented as a linear combination of these vectors. Pick $v_3 = [1,1,1,...1]$ as one of the vectors and two other column vectors $v_1, v_2$ such that all column vectors of $A$ are in the span of $v_1, v_2$ and $v_3$. Now, $A$ can be written as product of two matrices $A=L_{m\times 3}R_{3\times (n+1)}$, where the columns of $L$ are $v_1,v_2,v_3$. Each column $(a,b,c)$ of $R$ gives us the equation of a line $aX+bY+c=0$; then the arrangement of lines given by the first $n$ columns of $R$ has the same set of covectors as $\mathcal{P}$. Hence by Theorem \ref{thm:trt}, this line arrangement is equivalent to the pseudoline arrangement $\mathcal{P}$, and thus $\mathcal{P}$ is stretchable.
\end{proof}

{\bf Proof of Theorem \ref{thm:main}:}
Since stretchability of a pseudoline arrangement is equivalent to existential theory of reals \cite{shor}, the proof of this theorem follows directly from Lemma \ref{lemma:main}.

\subsection{Proof of Theorem~\ref{thm:strictmain}}

We now prove hardness for computing the minimum rank of {\spms}, a special case of {\gspms}.

To do this, we will study the covectors of {\em uniform} pseudoline arrangements in some more detail.
If the pseudoline arrangement is uniform then $|\mathcal{C}_0(\mathcal{P})| = {n \choose 2}$, $|\mathcal{C}_1(\mathcal{P})| = |\mathcal{C}_2(\mathcal{P})| = 1 + {n+1 \choose 2}$. By Theorem \ref{thm:trt}, the problem of stretchability of a pseudoline arrangement $\mathcal P$ can be translated into finding a line arrangement whose set
of covectors equals the set of covectors of $\mathcal P$ (up to some sign flipping). But for uniform pseudoline arrangements, we will see that we only need to find a line arrangement whose set of covectors corresponding to $2$-cells equals the set of covectors of $\mathcal P$ corresponding to $2$-cells (up to some sign flipping).

\begin{lemma}
\label{lemma:suff}
For uniform pseudoline arrangements $\mathcal{P}_1$ and $\mathcal{P}_2$, $\mathcal{P}_1$ and $\mathcal{P}_2$ are equivalent if $\mathcal{C}_2(\mathcal{P}_1) = \mathcal{C}_2(\mathcal{P}_2)$ (up to reorientation).
\end{lemma}
\begin{proof}
To prove this lemma, we will show - given the covectors in $\mathcal{C}_2$ of an uniform pseudoline arrangement, we can {\em uniquely} determine the set of covectors $\mathcal{C}_0$ and $\mathcal{C}_1$ associated with the arrangement. Hence by Theorem \ref{thm:trt}, the set of covectors $\mathcal{C}_2$ is enough to determine the equivalence class of a uniform pseudoline arrangement.

Let $\mathcal{P}$ be any uniform pseudoline arrangement with $n$ pseudolines. Note that every sign pattern in $\mathcal{C}_2(\mathcal{P})$ has full support. 

The idea to construct $\mathcal{C}_0(\mathcal{P})$ and $\mathcal{C}_1(\mathcal{P})$ given $\mathcal{C}_2(\mathcal{P})$ is simple (See figure~\ref{fig:covectors}): Suppose two sign patterns differ in only one position then the sign pattern with $0$ at that position, keeping other signs same as one of the sign vector, is in $\mathcal{C}_1(\mathcal{P})$. Similarly, if there are four sign vectors such that for a fixed two positions, every pair differs in those two positions, then the sign vector which has $0$ at these two positions, keeping other signs same as one of the sign vectors, is in $\mathcal{C}_0(\mathcal{P})$.

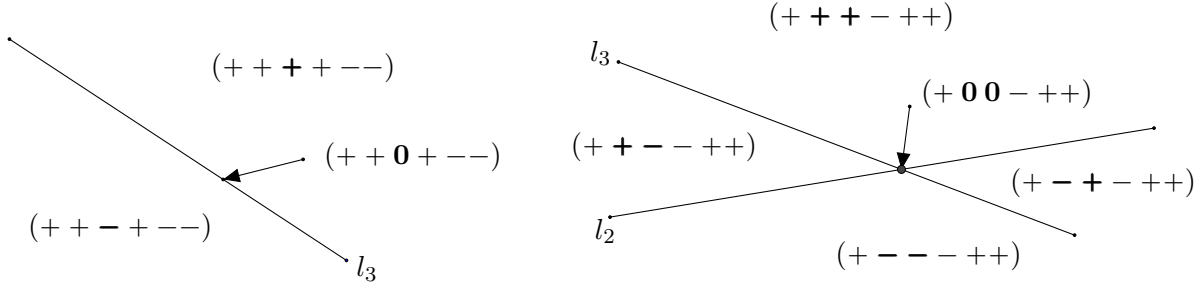
\begin{figure}[h!]
\label{fig:covectors}
\centering
\definecolor{uuuuuu}{rgb}{0.266666666667,0.266666666667,0.266666666667}
\definecolor{qqqqff}{rgb}{0.,0.,1.}
\begin{tikzpicture}[line cap=round,line join=round,>=triangle 45,x=0.8cm,y=0.8cm]
\clip(-2.4,-1.2) rectangle (21.,6.3);
\draw (3.52,0.66)-- (-2.08,4.34);
\draw (1.1,4.28) node[anchor=north west] {$(+ + \bm{+} + - - )$};
\draw (-1.94,1.62) node[anchor=north west] {$(+ + \bm{-} + - -)$};
\draw [->] (2.8,2.34) -- (1.47108323831,2.00643101482);
\draw (2.98,2.8) node[anchor=north west] {$(+ + \bm{0} + - -)$};
\draw (8.04,3.96)-- (15.62,1.08);
\draw (7.9,1.38)-- (16.94,2.86);
\draw (10.36,5.12) node[anchor=north west] {$(+ \bm{+ +} - + +)$};
\draw (7.08,2.96) node[anchor=north west] {$(+ \bm{+ -} - + +)$};
\draw (11.48,1.14) node[anchor=north west] {$(+ \bm{- -} - + +)$};
\draw (14.38,2.36) node[anchor=north west] {$(+ \bm{- +} - + +)$};
\draw [->] (12.88,3.22) -- (12.74341873,2.17294908412);
\draw (12.9,3.84) node[anchor=north west] {$(+\: \bm{0\: 0} - + +)$};
\draw (3.52,0.9) node[anchor=north west] {$l_3$};
\draw (7.48,1.52) node[anchor=north west] {$l_2$};
\draw (7.4,4.48) node[anchor=north west] {$l_3$};
\begin{scriptsize}
\draw [fill=qqqqff] (3.52,0.66) circle (0.5pt);
\draw [fill=black] (-2.08,4.34) circle (0.5pt);
\draw [fill=black] (2.8,2.34) circle (0.5pt);
\draw [fill=black] (1.47108323831,2.00643101482) circle (0.5pt);
\draw [fill=black] (8.04,3.96) circle (0.5pt);
\draw [fill=black] (15.62,1.08) circle (0.5pt);
\draw [fill=black] (7.9,1.38) circle (0.5pt);
\draw [fill=black] (16.94,2.86) circle (0.5pt);
\draw [fill=black] (12.88,3.22) circle (0.5pt);
\draw [fill=uuuuuu] (12.74341873,2.17294908412) circle (1.5pt);
\end{scriptsize}
\end{tikzpicture}
\caption{An example of deciding the sign pattern of the $1$ (left) and $0$ (right) cell complex given sign patterns of $2$ cell complexes.}
\end{figure}
More formally,
\begin{align*}
{\mathcal{C}_0^{'}} &= \left\{ \begin{array}{c|c} s \in \{-,0,+\}^n  &  \text{ } \exists \text{\it i, j s.t. }, i\neq j,  s_i = s_j = 0 \text{ and } \\ & \exists v, w \in {\mathcal{C}_2} \text{  s.t. } v_i \neq w_j \text{ and } v_j\neq w_i \text{ and } v_k=w_k=s_k, \forall k\neq i, j \end{array}\right\}\\
{\mathcal{C}_1^{'} } &= \{ s \in \{-,0,+\}^n  |  \text{ } \exists \text{ i  s.t. } s_i \in \{-,+\} \text{ and } \exists v\in {\mathcal{C}_0^{'}} \text{  s.t. } v_i=0 \text{ and } v_j=s_j,  \forall j\neq i \}\\
\end{align*}
It is easy to show that ${\mathcal{C}_0^{'}}$ and ${\mathcal{C}_1^{'}}$ contain exactly those sign vectors which are present in  $\mathcal{C}_0(\mathcal{P})$ and  $\mathcal{C}_1(\mathcal{P})$ respectively. Hence, the set of covectors $\mathcal{C}_2(\mathcal{P})$ is enough to determine the oriented matroid associated with uniform pseudoline arrangement.
\end{proof}

\begin{lemma}
\label{lemma:main2}
Let $\mathcal{P}$ be a {\em uniform} pseudoline arrangement.
Then $\mathcal{P}$ is stretchable if and only if $\minrank(S ) \leq 3$, where
$S = \mat(\mathcal C_2(\mathcal P)))$. 
Furthermore $S$ is a sign pattern matrix.
\end{lemma}
\begin{proof}
The proof is along the same line as the proof of lemma~\ref{lemma:main}.
Let $n$ be the number of lines in $\mathcal P$, $m = 1 + {n+1 \choose 2}$.
Note that $S$ is an $m \times n$ sign pattern matrix with no $0$ (which
follows from the fact that covectors in $\mathcal C_2(\mathcal P)$ have no $0$s).

We first prove that if the given pseudoline arrangement is stretchable, then $\minrank(\mat(S)) \leq 3$. If the given pseudoline arrangement is stretchable then there is a line arrangement $\mathcal L$ such that up to reorientation, $\mathcal C_2(\mathcal L) = \mathcal C_2(\mathcal P)$ (this is where we use Lemma~\ref{lemma:suff}). From this line arrangement, we can construct a matrix $L_{m\times 3}$ and $R_{3\times n}$ as before, where the rows of $L$ now correspond to points in the $2$-cells of the cell complex given by $\mathcal L$.  Consider $A=L\cdot R$, and note that
$\rank(A) \leq 3$. Since $\mathcal C_2(\mathcal L)$ equals $\mathcal C_2(\mathcal P)$ up to reorientation, 
and the $\sgn$s of the rows of $\mathcal A$ precisely equals $\mathcal C_2(\mathcal L)$,
we can flip the sign of some rows of $\mathcal A$ to get a matrix $A'$ with $\sgn(A') = S$ and $\rank(A') \leq 3$. Thus $\minrank(S) \leq 3$.

In the other direction, suppose $\minrank(\mat(S)) \leq 3$. Without loss of generality, assume that the first column of $S$ is all $+$'s (by flipping the signs of some rows). Now if $\minrank(\mat(S)) \leq 3$ then there exists some matrix $A \in \mathbb{R}^{m\times n}$ having rank at most $3$ and $\sgn(A) = \mat(S)$. By dividing every row $i$ of $A$ with $A_{(i,1)}$ we get the first column to equal the [1,1,1,...1] vector (since all entries in $A$ are non-zero). Since the column rank of $A$ is at most $3$, there exists $3$ vectors such that every column vector of $A$ can be represented as a linear combination of these vectors. Pick $v_3 = [1,1,1,...1]$ as one of the vector and two other column vectors $v_1, v_2$ such that all column vectors of $A$ are in the span of $v_1, v_2$ and $v_3$. Write $A$ as $A=L_{m\times 3}R_{3\times n}$, where columns of $L$ are $v_1,v_2,v_3$. If we consider each column $(a,b,c)$ of $R$ as an equation of line $ax+by+c=0$ then the set of covectors of 2-cell complexes formed by the $n$ lines exactly equals $\mathcal{C}_2(\mathcal{P})$. Hence by Lemma \ref{lemma:suff}, this lines arrangement is equivalent to the pseudoline arrangement $\mathcal{P}$. Thus $\mathcal{P}$ is stretchable.
\end{proof}

{\bf Proof of Theorem $\ref{thm:strictmain}$:}
Proof of this theorem follows from the $\NP$ hardness of stretchability of {\it uniform} pseudoline arrangement and Lemma \ref{lemma:main2}.


\section{A polynomial time algorithm for deciding if $\minrank \leq 2$}

In this section, we give a polynomial time algorithm to detect if the minimum rank of a given sign
pattern matrix is $\leq 2$.

We first give a simple combinatorial condition characterizing the {\spms}s with 
minimum rank $\leq 2$. It turns out that this combinatorial condition can be easily checked in polynomial time.
Next, we show how to reduce the problem for general sign pattern matrices to the problem for  {\spms} (this reduction only works for deciding $\minrank \leq 2$, not $\minrank \leq r$ for a general $r$).

\subsection{Sign pattern matrices}

We say a sign vector $s$ is contained in subspace $V$ if there exists a vector $v\in V$ such that $\sgn(v)=s$.
We say the subspace $V \subseteq \mathbb{R}^m$ {\em realizes} a {\spm} $S$ if all the columns of $S$ are contained in $V$.
Note that $\minrank(S)$ equals the smallest dimension of a subspace $V \subseteq \R^m$ that realizes $S$.

Our algorithm will be based on looking at certain set systems associated with {\spms}.
For sign vector $s$ of length $m$, let $T^-(s)$ denotes the set of $i\in [m]$ such that $s_i = -$. Given an $m \times n${\spm} $S$, define $T(S) :=\{ T^-(c_i) : 1\leq i\leq n\}$, where $c_i$ is the $i$th column of matrix $S$.

As a warmup, we start with a bound on the maximum number of distinct columns that a {\spm} of small $\minrank$ can have.
This follows from the next lemma, which gives an upper bound on the number of distinct sign patterns that are contained in a
subspace of dimension $d$ in $\mathbb{R}^m$ (we will eventually only use the $d = 2$ case of this).
\begin{obs}
\label{thm:maxsign}
Let $V$ be any subspace of dimension $d$ in $\mathbb{R}^m$. Then number $s(m,d)$ of different sign patterns (with no $0$) occurring in $V$ is at most
\begin{equation*}
s(m,d) \leq 2\left[ {m-1 \choose 0} + {m-1 \choose 1} +  {m-1 \choose 2} + \cdots +  {m-1 \choose d-1} \right]
\end{equation*}
\end{obs}
\begin{proof} 
Suppose not. Let $\mathfrak{F}$ be the set of all sign patterns contained in $V$; we have $|\mathfrak F| >  s(m,d)$.
Let $\mathfrak{F}'$ be the subset of $\mathfrak{F}$ whose last coordinate equals $+$; note that $|\mathfrak F'| = |\mathfrak F|/2$.
Let $T$ be the family of sets $ \{ T^{-}(c) \mid c \in \mathfrak F' \}$, and note that $|T| = |\mathfrak F'|$.

By the Sauer-Shelah Lemma \cite{Sauer, Shelah}, there exists a subset of coordinates $U \subseteq [m-1], |U| = d $ such that
$T|_{U} = 2^{U}$. Let $U' = U \cup \{m\}$. Then  $\mathfrak{F}|_{U'} = \{+, - \}^{d + 1}$.
But this cannot be, since there is some vector $w \in \R^m$ supported on coordinates in $U'$ which is orthogonal to all vectors in $V$,
and therefore $\sgn(w)|_{U'} \not\in \mathfrak{F}|_{U'}$.
\end{proof}

\begin{cor}
\label{cor:maxsign2}
A 2-dimensional subspace in $\mathbb{R}^m$ can have at most $2m$ distinct sign patterns with no $0$.
\end{cor}

We now come to the combinatorial property that will characterize when $\minrank \leq 2$.

\begin{df}
{\bf 2-Chain Property :} Let $T \subseteq 2^{[m]}$ with $|T|=2k$ for some $k>0$. 
We say that $T$ has the {\it 2-chain property}  there exist $A,B \subseteq T$ with {\it (i)} $A \cup B = T$, {\it (ii)} $|A| = |B| = k+1$,
and {\it (iii)} there exist orderings of elements of $A$ and $B$, $A = \{A_1,A_2,\dots,A_{k+1}\}$ and $B = \{B_1,B_2,\dots,B_{k+1}\}$, such that 
\begin{itemize}
\item $A_1 = \phi$
\item $A_{k+1} = [m]$
\item $A_i \subsetneq A_{i+1}$ , for all $1\leq i\leq k$
\item $B_i = \bar{A_i}$, for all $1\leq i \leq k+1$
\end{itemize}
{\bf Ex:} Consider $T\subseteq 2^{[3]}$ defined as $T= \{ \phi, \{1,2\}, \{1,2,3\}, \{3\}\}$, then $T$ has a $2$-chain property by setting $A_1 = \phi, A_2 = \{1,2\}$ and $A_3=\{1,2,3\}$.
\end{df}

\begin{df}
{\bf Complement of a sign vector : } Let $s$ be any sign vector in $\{-,0,+\}^n$, complement of $s$ is  $\bar{s}$ such that $\bar{s}_i = - s_i, 1\leq i\leq n$.
\end{df}

\begin{lemma}
\label{lemma:cp_eq_mr2}
If $S$ is an $m \times n$ {\spm}, such that:
\begin{enumerate}
\item all columns are distinct,
\item some column of $S$ is the column of all $+$s,
\item for every $i\in [n]$ there exists $j\in [n]$ such that $c_i = \bar{c_j}$ (where $c_i$ and $c_j$ are the sign vectors corresponding to columns $i$ and $j$ respectively).
\end{enumerate}
Then $\minrank(S) \leq 2$ if and only if $T(S)$ has the 2-chain property.
\end{lemma}
\begin{proof}
First suppose $\minrank(S) \leq 2$.
If $\minrank(S) = 1$, then $T(S) = \{\phi, [m]\}$, and so $T(S)$ trivially has the 2-chain property.
If $\minrank(S) = 2$. Then we can assume without loss of generality that there is a 2-dimensional space $V$ which realizes $S$
of the following special form: $V = \spn \{ X, Y\}$, where $X =[1,1,\dots,1]$ and $Y=[y_1,y_2,\dots,y_m]$, where $0<y_i < 1$ for all $i\in [m]$.
The only sign patterns appearing in the subspace spanned by $X$ and $Y$ are $F_1 = \{ \sgn(Y-\epsilon X) \mid \epsilon \in [0,1]\}$ 
and $F_2 = \{ \bar{s} \mid s \in F_1 \}$. When $\epsilon = 0$, $T^-(sign(Y)) = \phi$. Also, $T^-(sign(Y - \epsilon_1 Y)) \subseteq T^-(sign(Y - \epsilon_2 Y))$ for all $\epsilon_1,\epsilon_2$ where $\epsilon_2 > \epsilon_1$.  Finally, $T^-(Y-X) = [m]$. 
This shows that the set system $F$ defined by $F = \{T^{-}(v) \mid v \in V \}$ has the 2-chain property.
Since $T(S)$ is a subset of $F$ that is closed under complement (by condition 2), $T(S)$ also has the 2-chain property.

Now suppose $T(S)$ has the 2-chain property. We will construct vectors $X=[x_1,x_2,\dots,x_m]$ and $Y=[y_1,y_2,\dots,y_m]$ in $\mathbb{R}^m$ such that $V = \spn\{X,Y\}$ realizes $S$. Set $X=[1,1,\dots,1]$. Since $T(S)$ has the 2-chain property, there exist $A,B\subseteq T(S)$ of size $\frac{n}{2}+1$ as required by the definition
of the 2-chain property. Let $(A_1,A_2,\dots,A_{\frac{n}{2}+1})$ be the ordering of elements in $A$ such that $A_1 = \phi$, $A_{\frac{n}{2}+1} = [m]$ and $A_i \subset A_{i+1}$ , for all $1\leq i\leq \frac{n}{2}$. Set $y_j = \frac{1}{i}$ if $j\in A_i$ but $j\notin A_{i-1}$.
We now show that $V = \spn\{X, Y \}$ realizes $S$ (and hence $\minrank(S) = 2$). Let $c$ be a sign vector corresponding to some column of $S$. Either $T^-(c) = A_i$ or $T^-(\bar{c})= A_i$ for some $i\in [\frac{n}{2}+1]$. If $T^-(c) = A_i$  then $\sgn(X-(i+\frac{1}{2})Y) = c$ otherwise $\sgn((i+\frac{1}{2})Y-X) = c$. Hence the sign vector $c$ is contained in $V$. This completes the proof.
\end{proof}

\begin{lemma}
\label{lemma:cp_in_poly}
Given $S\subseteq 2^{[m]}$, we can check in polynomial time if $S$ satisfies the 2-chain property or not.
\end{lemma}
\begin{proof}
The conditions that $\phi,[m] \in S$, $|S|$ being even, and that for all $U \in S , [m]\setminus U \in S$ are easy to check. Suppose $S$ satisfies all these properties. Let $|S|=2k$ for some $k>0$. We will try to find two chains $A=(A_1,A_2,\dots,A_{k+1}),B=(B_1,B_2,\dots,B_{k+1})$ of size $k+1$ each such that $A\cup B = S$ and $B_i = \bar{A_i}$, for all $1\leq i \leq k+1$. Since we assumed that for all $s\in S , \bar{s}\in S$, the construction of chain $A$ will give us a chain $B$. We start with $A_1=\phi$. For $i=2$ to $k+1$, we will set $A_i$ to be the unique minimum sized set in $S\setminus\{A_1,\dots,A_{i-1}\}$ that contains $A_{i-1}$. We repeat this process
until we get $(A_1, \ldots, A_{k+1})$, and we verify that $A_{k+1} = [m]$.

There are two ways this process may fail.
Firstly, it could be that during iteration $i$, there are two minimum sized subsets in $S\backslash\{A_1,\dots,A_{i-1}\}$. In this case if we set $A_i$ to any one of them, then there is no place for the other one in either chain $A$ or chain $B$, and hence $S$ does not satisfy the 2-chain property. The other case is when the process sets $A_j=[n]$ for some $j<k+1$. In this case, the set $(S\backslash\{A_1,\dots,A_{j}\} )\cup \{B_{k+1},B_k,\dots,B_{k-j+2}\}$ is non empty, and this implies that $S$ does not have the 2-chain property.
\end{proof}
\begin{theorem}

\label{thm:effcomp}
There is a polynomial time algorithm to check if a given {\spm} $S$ satisfies $\minrank(S) \leq 2$.
\end{theorem}
\begin{proof}
Since the minimum rank is invariant under flipping signs of a row, we can convert $S$ to a {\spm} $S'$ whose first columns is all $+$s.
For every column $c$ of $S'$, if the complement $\bar{c}$ is not already a column of $S'$, we include it as a column into $S'$.
It is clear that $\minrank(S') = \minrank(S)$. The theorem follows from Lemma \ref{lemma:cp_eq_mr2} and Lemma~\ref{lemma:cp_in_poly}.
\end{proof}

\subsection{ Generalized sign pattern matrices}

We now show how to reduce the problem of detecting if $\minrank \leq 2$ for {\gspm}
to the same problem for {\spm}. 
This follows immediately from the following transformation.

\begin{theorem}
There is a polynomial time algorithm, which when given as input a {\gspm} $S$, either declares that $\minrank(S) > 2$, or else 
outputs a sign pattern matrix $S'$ such that $\minrank(S') = \minrank(S)$.
\end{theorem}
\begin{proof}
The algorithm will work by applying a sequence of transformations
to $S$ that
will gradually reduce the number of $0$s in $S$. None of these
transformations change $\minrank(S)$. However, it may be that one of the transformations of the algorithm fails,
in which case we will be able to certify that $\minrank(S) > 2$.

{\bf Step 1:}
First remove every all-$0$ row and all-$0$ column from $S$.

{\bf Step 2:}
Now suppose $S$ is an $m \times n$ matrices.
If $S$ has $0$, then we may permute columns
and assume that the first column of $S$
has at least one $0$. 
Now by permuting rows and possibly flipping signs of rows,
we may assume that the first 
column of $S$ is a sequence of $k$ $0$s followed by a sequence of $m-k$, $+$s, 
where $0 < k < m$.

Now if $\minrank(S) \leq 2$, 
there must be an $A$ such that $\rank(A) \leq 2$
and $\sgn(A) = S$.
By scaling the rows of $A$,
we may assume that the first column of $A$
is the vector $X = (0,0, \ldots, 0, +1, +1, \ldots, +1)$, 
where there are $k$ $0$s and $m-k$ $+1$s.

Since the first row is not identically $0$, we can take another column $Y$ of
$A$ whose first coordinate is nonzero. Note that $X$ and $Y$ must be linearly independent.
Thus, every column of $A$ is of the form
$\alpha X + \beta Y$. In particular, since $A$ has no all-$0$ rows, 
all the first $k$ coordinates of $Y$ must be nonzero.

{\bf Step 3:} Pick a column of $S$ whose first coordinate is nonzero (i.e. either $+$ or $-$); if 
all the first $k$ coordinates of this column are not nonzero, declare
that $\minrank(S) > 2$. Otherwise, flip signs of the first $a$ rows of $S$ so that
the first $k$ coordinates of that column are all $+$.

Now by scaling the first $k$ rows of $A$,
we may assume that the first $k$ coordinates of $Y$ equal $+1$.

Suppose $Y = (1,1, \ldots, 1, y_{k+1}, \ldots, y_m)$.
Suppose some column $i \in [n]$ of $S$ has at least two $0$s (if not, goto Step $5$).
Let $U \subseteq [m]$ be the indices of all the $0$ coordinates of that column.
Note that we must have
either $U = [k]$ or $U \subseteq [m] \setminus [k]$.
In either case, we have that all the coordinates of $Y$ indexed by $U$
must be equal. Thus all the rows of $A$ indexed by $U$ are identical
to one another. 

{\bf Step 4:} Pick a column of $S$ with at least two $0$s, and let $U \subseteq [m]$ be the indices
of the $0$ coordinates. If $U \neq [k]$ and $U \not\subseteq [m] \setminus [k]$, then
declare that $\minrank(S) > 2$. Otherwise, if all the rows with indices in $U$
are not identical, declare that $\minrank(S) > 2$. Otherwise, delete all but one of these rows from $S$.
Repeat this until $S$ has no column with more than one $0$.

Our algorithm therefore may check if all the rows of $S$ indexed by $U$ 
are identical to one another, and if so, removes all but one of them.
Otherwise it declares that $\minrank(S) > 2$.
In the analysis, we also perform this duplicate removal on $A$, which clearly preserves the rank.

{\bf Step 5:}
Now, take any column $C$ of $S$ with exactly one $0$.
Replace $C$ with two columns $C_+$ and $C_-$, with $C_+$ being a copy of $C$ with the $0$ replaced by $+$, and $C_-$ being a copy of $C$ with the $0$ replaced by $-$.
Repeat this until $S$ has no $0$s.

We need to show that this operation does not change $\minrank(S)$.
In one direction, if $V$ is a subspace of $\mathbb R^m$ of dimension $d$ such that
some two elements $p_+, p_-$ of $V$ have sign patterns $C_+$, $C_-$ respectively, then there must be an element of $V$ with sign pattern $C$
(by considering the line between $p_+$ and $p_-$). In the other direction, if $V$ is a subspace of $\mathbb R^m$ of dimension $d$
which does not have any identically $0$ coordinates, and $p \in V$ has sign pattern $C$, then infinitesimal perturbations
of $p$ within $V$ will have sign patterns $C_+$ and $C_-$.

Thus, if we had not already declared that $\minrank(S) > 2$,
we get a sign pattern matrix $S'$ with $\minrank(S) = \minrank(S')$. 
\end{proof}


\section{Open problems}
We conclude with some interesting open problems on minimum ranks of sign pattern matrices.
\begin{enumerate}
\item Can we compute the minimum rank of a sign pattern matrices up to an $O(1)$ additive error? We only rule out such $O(1)$ additive error polynomial time algorithms for computing the minimum rank of a {\gspm} (unless $\PTIME = \NP$).
\item Is it possible to approximate the minimum rank of a {\gspm} within some constant factor? We showed that this constant factor must be $> 4/3$.

Very recently, Alon, Moran and Yehudayoff~\cite{AMY14} gave an $O(n/\log n)$-factor approximation algorithm for the minimum rank of a sign pattern matrix.
\item How large can the minimum rank of a $n\times n$ sign pattern matrix be? We know it is $\Theta(n)$, but the precise constant in front of the $n$ is not known.
\item Can we construct explicit $n\times n$ sign pattern matrices whose minimum rank is $\omega(\sqrt{n})$?
\end{enumerate}

\nocite{*}
\bibliographystyle{alpha}	
\bibliography{refs}		

\end{document}